\documentclass[11pt]{article}
\usepackage[a4paper, margin=1in]{geometry}
\usepackage{amsthm}

\newtheorem{theorem}{Theorem}
\newtheorem{lemma}[theorem]{Lemma}

\begin{document}

\title{Computation with Large Advice}

\author{Hiroki Morizumi\\
{\small Shimane University, Japan}\\
{\small morizumi@cis.shimane-u.ac.jp}
}

\date{}

\maketitle

\begin{abstract}
In this paper, we consider a new direction of computation, which we
call computation with large advice.
We mainly consider constant space computation with large advice
in Turing machines, and prove the following facts:
(i) The class of decision problems solvable by a constant
space Turing machine with polynomial-size advice includes
nonuniform-{\sf NC}$^1$,
(ii) The class of decision problems solvable by a constant
space Turing machine with quasipolynomial-size advice equals
nonuniform-{\sf polyL}.
The facts mean constant space computation
with large advice has unexpected computational power.
On the other hand, we mention bounded time computation with large advice,
and attempt to propose a concept of ``algorithms with large advice''.
In the proposal, advice is precomputed data for a problem and a fixed
instance size, and we expect efficient algorithms by large or huge advice.
\end{abstract}

\section{Introduction}

In this paper, we prove the following theorems.

\begin{theorem} \label{thrm:poly}
The class of decision problems solvable by a constant
space Turing machine with polynomial-size (sequential access) advice
includes nonuniform-{\sf NC}$^1$.
\end{theorem}

\begin{theorem} \label{thrm:qpoly}
The class of decision problems solvable by a constant
space Turing machine with quasipolynomial-size (sequential access) advice
equals nonuniform-{\sf polyL}.
\end{theorem}

\noindent Two theorems imply unexpected power of large advice
in computation.
We consider ``computation with large advice'' in this paper.

\subsection{Advice}

Advice of Turing machines is an extra input which depends only
on the size of the main input.
Advice is given on a read-only tape of sequential access in this paper.
Advice is closely related to the nonuniformity of complexity classes.
For example, the nonuniform variant of {\sf P} is {\sf P/poly} which is
the class of decision problems solvable by a polynomial time Turing machine
with polynomial size advice.
The nonuniform variant of {\sf L} is {\sf L/poly} which is
the class of decision problems solvable by a logarithmic space Turing machine
with polynomial size advice.

In {\sf P/poly} and {\sf L/poly}, advice gives no computational power
beyond nonuniformity.
In this paper, we consider large advice which gives extra power
beyond nonuniformity.

\subsection{Implications of Theorem~\ref{thrm:qpoly}}

Nonuniform-{\sf polyL} ($=$ {\sf polyL/quasipoly}) includes various
complexity classes, e.g. as follows.
\begin{itemize}
\item {\sf L} with quasipolynomial-size (sequential access) advice
\footnote{If advice is large advice, we do not use notations such as
{\sf L/quasipoly} and {\sf NL/quasipoly} in this paper.
For large advice, known facts for advice is not necessarily valid and
further studies are needed.}
\item {\sf NL} with quasipolynomial-size (sequential access) advice
\end{itemize}
\noindent The classes above includes the class of decision problems solvable
by a constant space Turing machine with quasipolynomial-size
(sequential access) advice.
Furthermore, nonuniform-{\sf polyL} equals
\begin{itemize}
\item nonuniform-{\sf NC} of quasipolynomial size (i.e., the class of
decision problems solvable by circuits with quasipolynomial size,
depth $\log^{O(1)} n$, and fan-in 2)
\end{itemize}
Thus, by Theorem~\ref{thrm:qpoly}, all complexity classes above
are equivalent.

\subsection{The Contents of The Paper}

The main results of this paper are Theorem~\ref{thrm:poly} and
Theorem~\ref{thrm:qpoly}.
In Section~\ref{sec:idea}, we describe our ideas and contribution.
In Section~\ref{sec:poly} and Section~\ref{sec:qpoly}, we prove
Theorem~\ref{thrm:poly} and Theorem~\ref{thrm:qpoly}, respectively.

Theorem~\ref{thrm:poly} and Theorem~\ref{thrm:qpoly} raise questions and
interests for the power of large advice in computation.
In Section~\ref{sec:conc}, we discuss further studies of large advice.
Especially, we mention the power of large advice for algorithms.
Although the contents of Section~\ref{sec:conc} are conceptual,
the discussion strengthens the motivation of this paper.

\medskip

\noindent {\bf Note.}
In Theorem~\ref{thrm:poly} and Theorem~\ref{thrm:qpoly}, we assume that
advice is sequential access, and the theorems strongly depend on
sequential access.
On the other hand, we also discuss random access advice for
future works in Section~\ref{sec:conc}.

\section{Preliminaries}

\subsection{Definitions}

\subsubsection{Branching programs and Boolean circuits}

A {\em branching program} is a directed acyclic graph.
The nodes of out-degree 2 are called {\em inner nodes} and labeled
by a variable.
The nodes of out-degree 0 are called {\em sinks} and labeled by 0 or 1.
For each inner node, one of the outgoing edges is labeled by 0 and
the other one is labeled by 1.
There is a single specific node called the {\em start node}.
An assignment to the variables determines a computation path from
the start node to a sink node.
The value of the sink node is the output of the branching program.
If the nodes are arranged into a sequence of levels with edges going only
from one level to the next, then the {\em width} is the size of
the largest level.

{\em Circuits} are formally defined as directed acyclic graphs.
The nodes of in-degree 0 are called {\em inputs}, and each one of them
is labeled by a variable or by a constant 0 or 1.
The other nodes are called {\em gates}, and each one of them
is labeled by a Boolean function.
The {\em fan-in} of a node is the in-degree of the node, and
the {\em fan-out} of a node is the out-degree of the node.
In this paper, the gates are AND gates of fan-in two, OR gates of fan-in two,
and NOT gates.
There is a single specific node called {\em output}.
The {\em size} of a circuit is the number of gates in the circuit.
The {\em depth} of a circuit is the length of the longest path in the circuit.

\subsubsection{Complexity classes}

{\sf L} is the class of decision problems solvable by a logarithmic
space Turing machine.
{\sf NL} is the nondeterministic variant of {\sf L}.
{\sf polyL} is the class of decision problems solvable by a polylogarithmic
space Turing machine.
{\sf polyL/quasipoly} is the class of decision problems solvable by
a polylogarithmic space Turing machine with quasipolynomial size advice.

Let $n$ be the number of inputs in circuits.
{\sf NC}$^1$ is the class of decision problems
solvable by a uniform family of circuits with polynomial size,
depth $O(\log n)$, and fan-in 2.
{\sf NC} is the class of decision problems
solvable by a uniform family of circuits with polynomial size,
depth $\log^{O(1)} n$, and fan-in 2.

\subsubsection{The sorting problem}

The sorting problem is defined as follows in this paper.

\medskip

\noindent {\bf Input:} $n$ numbers in the set $\{0, 1, \ldots, k\}$.

\smallskip

\noindent {\bf Output:} A permutation $\langle a_1, a_2, \ldots, a_n \rangle$
of the input such that $a_1 \leq a_2 \leq \cdots \leq a_n$.

\medskip

\noindent Note that $n$ numbers are restricted in
the set $\{0, 1, \ldots, k\}$.

\subsection{A Conversion Lemma}

We use the following lemma in the proof of Theorem~\ref{thrm:qpoly}
(in Section~\ref{sec:qpoly}).

\begin{lemma} \label{lem:conv}
Any quasipolynomial-size branching program can be converted to
a circuit of polylogarithmic depth.
\end{lemma}

\noindent Lemma~\ref{lem:conv} is proved by the following lemma.

\begin{lemma} \label{lem:conv2}
If any branching programs of size $s$ can be converted
to a circuit of depth $d$,
then any branching programs of size $2s$ can be converted
to a circuit of depth $d + \lceil \log s^2 \rceil + 2$.
\end{lemma}
\begin{proof}
Let $G$ be a branching programs of size $2s$.
Let $G_1$ and $G_2$ be the former $s$ nodes and the latter $s$ nodes,
respectively, in arbitrary topological sorted order.
Let $E_1$ be the edges between $G_1$ and $G_2$.
The number of edges in $E_1$ is at most $s^2$.
All paths from the start node to a sink node contain one edge in $E_1$.
For each edge in $E_1$, we check the existence of a path from the start
node to the 1 sink node.
Natural construction of such circuit is enough to prove the lemma.
\end{proof}

\begin{proof}[Proof of Lemma~\ref{lem:conv}]
We apply Lemma~\ref{lem:conv2} recursively.
\end{proof}

\section{The Key Ideas and Our Contribution} \label{sec:idea}

Our results of this paper are closely related to Barrington's theorem
which has been well-known in computational complexity theory.
In the 1980's, Barrington proved the following theorem.

\begin{theorem}[\cite{B89}] \label{thrm:bar}
The class of decision problems solvable by a nonuniform family of
polynomial-size $5$-width branching programs equals
nonuniform-{\sf NC}$^1$.
\end{theorem}

\noindent The theorem is a surprising result, since
$5$-width branching programs have only at most five states in each stage
of computation. Thus, it has been known that computation with constant
working space can have unexpected computational power.
In this paper, we represent the computation in Barrington's theorem
by Turing machines with large advice, which is our main contribution and
gives a new insight to computation.

\section{Proof of Theorem~\ref{thrm:poly}} \label{sec:poly}

We firstly prove a relation between constant space Turing machines
with polynomial-size advice and
polynomial-size $5$-width branching programs.

\begin{lemma} \label{lem:simu-p}
Constant space Turing machines with polynomial-size advice simulate
polynomial-size $5$-width branching programs.
\end{lemma}
\begin{proof}
We prepare an encoding of the branching program on the tape of advice.
The execution of the branching program is simulated with sequential access
to the encoding.
(The Turing machine does not store the position of the advice tape.
The tape head of advice moves to only one direction.)
The encoding for variables of branching programs needs to be careful.
The encoding for variable $x_i$ consists of one starting symbol and
$n$ symbols, the $i$'th of which is marked.
\end{proof}

\noindent We prove Theorem~\ref{thrm:poly}.

\begin{proof}[Proof of Theorem~\ref{thrm:poly}]
By Theorem~\ref{thrm:bar} and Lemma~\ref{lem:simu-p}, the theorem holds.
\end{proof}

\section{Proof of Theorem~\ref{thrm:qpoly}} \label{sec:qpoly}

The proof outline of Theorem~\ref{thrm:qpoly} is similar to the one of
Theorem~\ref{thrm:poly}.
To prove Theorem~\ref{thrm:qpoly}, we prove
the quasipolynomial-size variant of Barrington's theorem.

\subsection{The Quasipolynomial-Size Variant of Barrington's Theorem}

The following theorem is the quasipolynomial-size variant of
Barrington's theorem.
(Theorem~\ref{thrm:q-bar} may be of independent interest as a variant of
Barrington's theorem, although the theorem is not the main aim of the paper.)

\begin{theorem} \label{thrm:q-bar}
The class of decision problems solvable by a nonuniform family of
quasipolynomial-size $5$-width branching programs equals
nonuniform-{\sf polyL}.
\end{theorem}

To prove Barrington's theorem, Barrington has proved the following lemma.
The lemma is useful also for the proof of Theorem~\ref{thrm:q-bar}.
(We have added some modifications to fit the paper.)

\begin{lemma}[Theorem~1 of \cite{B89} (with some modifications)] \label{lem:bar}
Any circuit with depth $d$ and fan-in 2 can be converted to
a $5$-width branching program with length at most $4^d$.
\end{lemma}

\noindent We prove Theorem~\ref{thrm:q-bar}.

\begin{proof}[Proof of Theorem~\ref{thrm:q-bar}]
By Lemma~\ref{lem:conv} and Lemma~\ref{lem:bar},
any quasipolynomial-size branching program can be converted to
a quasipolynomial-size $5$-width branching program.
Since the class of decision problems solvable by a nonuniform family of
quasipolynomial-size branching programs equals nonuniform-{\sf polyL},
the theorem holds.
\end{proof}

\subsection{Proof of Theorem~\ref{thrm:qpoly}}

This subsection is almost similar to Section~\ref{sec:poly}.

\begin{lemma} \label{lem:simu-qp}
Constant space Turing machines with quasipolynomial-size advice simulate
quasipolynomial-size $5$-width branching programs.
\end{lemma}
\begin{proof}
We can prove the lemma by the similar way of Lemma~\ref{lem:simu-p}.
\end{proof}

\noindent We prove Theorem~\ref{thrm:qpoly}.

\begin{proof}[Proof of Theorem~\ref{thrm:qpoly}]
Nonuniform-{\sf polyL} obviously includes
the class of decision problems solvable by a constant
space Turing machine with quasipolynomial-size advice.
By Theorem~\ref{thrm:q-bar} and Lemma~\ref{lem:simu-qp}, the theorem holds.
\end{proof}

\section{Towards Further Studies of Large Advice} \label{sec:conc}

The main results of the paper (Theorem~\ref{thrm:poly} and
Theorem~\ref{thrm:qpoly}) raise questions and
interests for the power of large advice in computation.
In this section, we discuss further studies for computation with large advice.

\subsection{Bounded Time Computation and Large Advice}

The power of large advice in bounded time computation depends on access to
advice.
If advice is sequential access, advice beyond polynomial size is obviously
useless for polynomial time computation since we cannot access superpolynomial
advice in polynomial time.
However, if advice is random access, then advice is useful.
For example, advice of exponential size can have acceptances of all inputs.
Even if {\sf P} $\neq$ {\sf NP}, we can resolve the problems in {\sf NP}
if {\sf NP} is included in {\sf P} with advice of reasonable size and
we can prepare the advice.
The observation above raises an idea of algorithms with large advice,
which is discussed in the next subsection.

\subsection{Algorithms with Large Advice} \label{subsec:algo}

In this subsection, we attempt to propose a concept of
``algorithms with large advice''.
In our proposal, advice is precomputed data for a problem and a fixed
instance size, and we expect efficient algorithms by large or huge advice.
(We call the data advice, since the data correspond to advice of Turing
 machines.)
In the following example, the merge sort algorithm for the sorting problem
accelerates from $O(n \log n)$ time to $O(n \log\log n)$ time by huge advice.

\begin{theorem}
A modified merge sort algorithm runs in $O(n \log\log n)$ time, if
$n$ is fixed and $\tilde{O}(k^{\frac{n}{\log n}})$ advice (i.e., precomputed data
of $\tilde{O}(k^{\frac{n}{\log n}})$ space) is given.
\end{theorem}

\begin{proof}
We execute the $\log\log n$ usual recursive steps of the merge sort, and
sort $\frac{n}{\log n}$ numbers by precomputed data.
Precomputed data has all patterns of sequences of $\frac{n}{\log n}$ numbers
and the sorted sequences.
\end{proof}

The sorting problem is an example for the framework of algorithms with
large advice.
We are interested in algorithms for computationally hard problems
such as {\sf NP}-hard problems, e.g., the SAT problem.

\bigskip

\noindent {\bf Remark 1.}
In one of situations which we suppose, advice (i.e., precomputed data) is
stored in some places of the internet as public data.
The current developed internet offers good environment for algorithms
with large advice.

\medskip

\noindent {\bf Remark 2.}
In this paper, we described a preliminary and minimal framework for
algorithms with large advice.
For actual execution, we need an efficient algorithm to compute advice
(i.e., precomputed data), although the algorithm runs once and the obtained
advice can be stored.

\bibliographystyle{plain}
\bibliography{circuit}

\end{document}